\documentclass[preprint,10pt]{elsarticle}

\usepackage{amsmath, amsthm, amssymb}
\usepackage{lipsum}
\usepackage[font=footnotesize, labelfont=bf]{caption}
\usepackage{listings}
\usepackage{gensymb}
\usepackage{fancyhdr}
\usepackage{graphicx}

\newtheorem{theorem}{Theorem}[section]

\journal{Parallel Computing}
\begin{document}
\begin{frontmatter}

\title{A Non-linear GPU Thread Map for Triangular Domains}

\author[uach]{Crist\'obal A. Navarro\corref{cor1}}
\ead{cnavarro@inf.uach.cl}
\author[dcc]{Benjam\'in Bustos}
\author[dcc]{Nancy Hitschfeld}
\address[uach]{Instituto de Inform\'atica, Universidad Austral de Chile.}
\address[dcc]{Departamento de Ciencias de la Computaci\'on, Universidad de Chile.}
\cortext[cor1]{Corresponding author}

\begin{abstract}
%\boldmath
There is a stage in the GPU computing pipeline where a grid of thread-blocks, in
\textit{parallel space}, is mapped onto the problem domain, in \textit{data
space}.  Since the parallel space is restricted to a box type geometry, the
mapping approach is typically a $k$-dimensional bounding box (BB) that covers a
$p$-dimensional data space. Threads that fall inside the domain perform
computations while threads that fall outside are discarded at runtime.  In this
work we study the case of mapping threads efficiently onto triangular domain
problems and propose a block-space linear map $\lambda(\omega)$, based on the
properties of the lower triangular matrix, that reduces the number of
unnnecessary threads from $\mathcal{O}(n^2)$ to $\mathcal{O}(n)$.  Performance
results for global memory accesses show an improvement of up to $18\%$ with
respect to the \textit{bounding-box} approach, placing $\lambda(\omega)$ on second
place below the \textit{rectangular-box} approach and above the
\textit{recursive-partition} and \textit{upper-triangular} approaches. For
shared memory scenarios $\lambda(\omega)$ was the fastest approach achieving $7\%$ of
performance improvement while preserving thread locality. The results obtained
in this work make $\lambda(\omega)$ an interesting map for efficient GPU computing on
parallel problems that define a triangular domain with or without neighborhood
interactions. The extension to tetrahedral domains is analyzed, with
applications to triplet-interaction n-body applications. 
\end{abstract}

\begin{keyword}
block-space mapping \sep data re-organization \sep trianglular domain %%
keywords here, in the form: keyword \sep keyword
\end{keyword}
\end{frontmatter}

\section{Introduction}
\label{sec_introduction}
GPU computing has become a well established research area \cite{4490127,
Nickolls:2010:GCE:1803935.1804055, navhitmat2014} since the release of
programable graphics hardware and its programming platforms such as CUDA
\cite{nvidia_cuda_guide} and OpenCL \cite{opencl08}.  In the CUDA GPU
programming model there are three constructs\footnote{OpenCL chooses different
names for these constructs; (1) work-element, (2) work-group and (3) work-space,
respectively.} that allow the execution of highly parallel algorithms; (1)
thread, (2) block and (3) grid.  Threads are the smallest elements and they are
in charge of executing the instructions of the GPU kernel. A block is an
intermediate structure that contains a set of threads organized as an Euclidean
box.  Blocks provide fast shared memory access as well as local synchronization
for all of its threads. The grid is the largest construct of all three and it
keeps all blocks together spatially organized for the execution of a
GPU kernel.  These three constructs play an important role when mapping the
execution resources to the problem domain.
%and are necessary for the GPU to
%schedule and distribute the work properly among its clusters of processing
%cores.  

For every GPU computation there is a stage where threads are mapped 
onto the problem domain. A map, defined as $f: \mathbb{R}^k \rightarrow
\mathbb{R}^p$, transforms each $k$-dimensional point $x=(x_1, x_2, ..., x_k)$ of
the grid into a unique $p$-dimensional point $f(x) = (y_1, y_2, \cdots, y_p)$ of
the problem domain. Since the grid lives in \textit{parallel space}, we
have that $f(x)$ \textit{maps the parallel space onto the data space}.  In
GPU computing, a typical mapping approach is to build a \textit{bounding-box} (BB)
type of parallel space, sufficiently large to cover the data space and map
threads using the identity $f(x) = x$. Such map is highly efficient for
the class of problems where data space is defined by a $p$-dimensional box, such
as vectors, lists, matrices and box-shaped volumes. For this reason, the
identity map is typically considered as the default approach.

%There is a large list of parallel problems have a box shaped data space;
%\textit{e.g.,} matrix multiplication, vector arithmetic, classic tiled
%computations and image processing kernels, among many others.  For such cases, a
%bounding-box (BB) parallel space is a natural approach as well as the optimal
%one since it provides fine-grained parallelism to the entire problem domain
%through the identity map $f(x) = x$, with a marginal number of unnecessary
%threads. For this class of problems the mapping process is not an issue.

There is a class of parallel problems where data space follows a triangular
organization. Problems such as the \textit{Euclidean distance maps} (EDM)
\cite{5695222, Li:2010:CME:1955604.1956601, Man:2011:GIC:2117688.2118809},
\textit{collision detection} \cite{AvrilGA12}, adjacency matrices
\cite{kepner2011graph}, \textit{cellular automata simulation on triangular
domains} \cite{ConwaysLife}, matrix inversion
\cite{Ries:2009:TMI:1654059.1654069}, \textit{LU/Cholesky decomposition}
\cite{springerlink_gustavson} and the \textit{n-body} problem
\cite{DBLP:journals/corr/abs-1108-5815, Bedorf:2012:SOG:2133856.2134140, 
Ivanov:2007:NPT:1231091.1231100}, among others, belong to this
class and are frequently encountered in the fields of science and
technology. In this class, data space has a size of $D=n(n+1)/2 \in
\mathcal{O}(n^2)$ and is organized in a triangular way. This data shape
makes the default \textit{bounding-box} (BB) approach inefficient as it
generates $n(n-1)/2 \in \mathcal{O}(n^2)$ unnesessary threads (see Figure
\ref{fig_bb_strategy}).
\begin{figure}[ht!]
\centering
\includegraphics[scale=0.18]{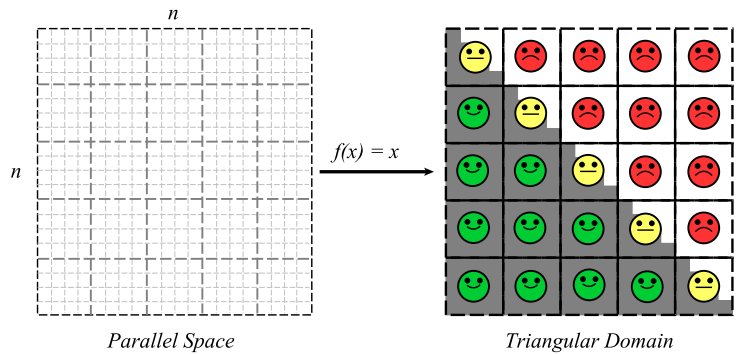}
\caption{The BB strategy is not the best choice for a \textit{td-problem}.}
\label{fig_bb_strategy}
\end{figure}

For each unnecessary thread, a set of instructions must be executed in order to
discard themselves from the useful work, leading to a performance penalty that
can become notorious considering that fine-grained parallelism usually produces
thousands of threads with a small amount of work per thread. In this context, it
is interesting to study how one can reduce the number of unnecessary threads to
a marginal number and eventually produce a performance improvement for all
problems in this class.  Throughout the paper, this class will be referred as
the \textit{triangular-domain} class or simply \textit{TD}. 

This work addresses the problem depicted in Figure \ref{fig_bb_strategy} and
proposes a block-space map $\lambda(\omega)$ designed for the TD class that reduces
the number of unnecessary threads to $o(n^2)$,
making kernel execution up to $18\%$ faster than the BB approach in global
memory scenarios and up to $7\%$ faster in shared memory tiled scenarios. The
rest of the paper includes a description of the related works (Section
\ref{sec_relatedwork}), a formal definition for $\lambda(\omega)$ (Section
\ref{sec_td_strategy}), an evaluation of square root implementations (Section
\ref{sec_implementation}), a performance comparison with the related works
(Section \ref{sec_experimental_results}) and an extension of the strategy to
tetrahedral domains with its potential performance benefit (Section \ref{sec_extension3d}). 

\section{Related Work}
\label{sec_relatedwork}
In the field of distance maps, Ying \textit{et. al.} have proposed a GPU
implementation for parallel computation of DNA sequence distances
\cite{Ying:2011:GDD:2065356.2065583} which is based on the Euclidean distance
maps (EDM). The authors mention that the problem domain is indeed
symmetric and they do realize that only the upper or lower triangular part of
the interaction matrix requires computation. Li \textit{et. al.}
\cite{Li:2010:CME:1955604.1956601} have also worked on GPU-based EDMs on large
data and have also identified the symmetry involved in the computation.
 
Jung \textit{et. al.} \cite{Jung2008} proposed packed data structures
for representing triangular and symmetric matrices with applications to LU and
Cholesky decomposition \cite{springerlink_gustavson}. The strategy is based on
building a \textit{rectangular box strategy} (RB) for accessing and storing a
triangular matrix (upper or lower). Data structures become practically half the
size with respect to classical methods based on the full matrix. The strategy
was originally intended to modify the data space (\textit{i.e.,} the matrix),
however one can apply the same concept to the parallel space.

Ries \textit{et. al.} contributed with a parallel GPU method for the
triangular matrix inversion \cite{Ries:2009:TMI:1654059.1654069}.  The authors
identify that the parallel space indeed can be improved by using a
\textit{recursive partition} (REC) of the grid, based on a \textit{divide and
conquer} strategy.

Q. Avril \textit{et. al.} proposed a GPU mapping function for collision
detection based on the properties of the \textit{upper-triangular map}
\cite{AvrilGA12}. The map, referred here as \textit{UTM}, is a thread-space function $u(x)
\rightarrow (a, b)$, where $x$ is the linear index of a thread $t_x$ and the
pair $(a,b)$ is a unique two-dimensional coordinate in the upper triangular
matrix. Their map is accurate in the domain $x \in [0, 100M]$, with a range of
$(a,b) \in [0,3000]$.

The present work is an extended and improved version of a previous conference
research by Navarro and Hitschfeld \cite{DBLP:conf/hpcc/NavarroH14}.

\section{Block-space triangular map}
\label{sec_td_strategy}
\subsection{Formulation}
\label{sec_block_mapping_function}
It is important to distinguish between two possible approaches; (1) thread-space
mapping and (2) block-space mapping.  Thread-space mapping is where each thread
uses its own unique coordinate as the parameter for the mapping function. The
approach has been used before in the work of Avril \textit{et. al.}
\cite{AvrilGA12}. 
On the other hand, block-space mapping uses the shared block coordinate to map
to a specific location, followed by a local offset on each thread according to
the relative position in their block. This approach has not been considered by
the earlier works and it has been chosen as it can give certain advantages over
thread-space mapping specially on memory access patterns.

For a problem in the TD class of linear size $n$, its total size is
$n(n+1)/2$. Let $m = \lceil n/\rho\rceil$ be the number of blocks needed to
cover the data space along one dimension and $\rho$ the number of threads per block
per dimension, or \textit{dimensional block-size} (for simplicity, we assume a
regular block).  A \textit{bounding-box} mapping approach would build a
box-shaped parallel space, namely $P_{BB}$, of $m \times m$ blocks and put
conditional instructions to cancel the computations outside the problem domain.
Although the approach is correct, it is inefficient since $m(m+1)/2$ blocks are
already sufficient when organized as:
\begin{equation}
D = 
\begin{vmatrix}
0	&		&		&				&					\\
1	&	2	&		&				&					\\
3	&	4	&	5	&				&					\\	
...	&	... 	&	...	&	...			& 					\\
\frac{m(m-1)}{2}	&	\frac{m(m-1)}{2} + 1 	&	...	&	...	& \frac{m(m+1)}{2}-1\\
\end{vmatrix}
\label{eq_matrix_indexing}
\end{equation}
We define a balanced two-dimensional parallel space $P_{\Delta}$  (see Figure
\ref{fig_ltm_strategy}) of size of $m' \times m'$ blocks
where $m' = \Big\lceil\sqrt{m(m+1)/2} \Big\rceil$. This setup reduces the number
of unnecessary threads from $n(n-1)/2 \in \mathcal{O}(n^2)$ to $\frac{\rho(\rho
- 1)}{2}\lceil n/\rho \rceil < \frac{\rho^2}{2} \lceil n/\rho \rceil \in o(n^2)$
threads, with $\rho \in O(1)$.
\begin{figure}[ht!]
\centering
\includegraphics[scale=0.18]{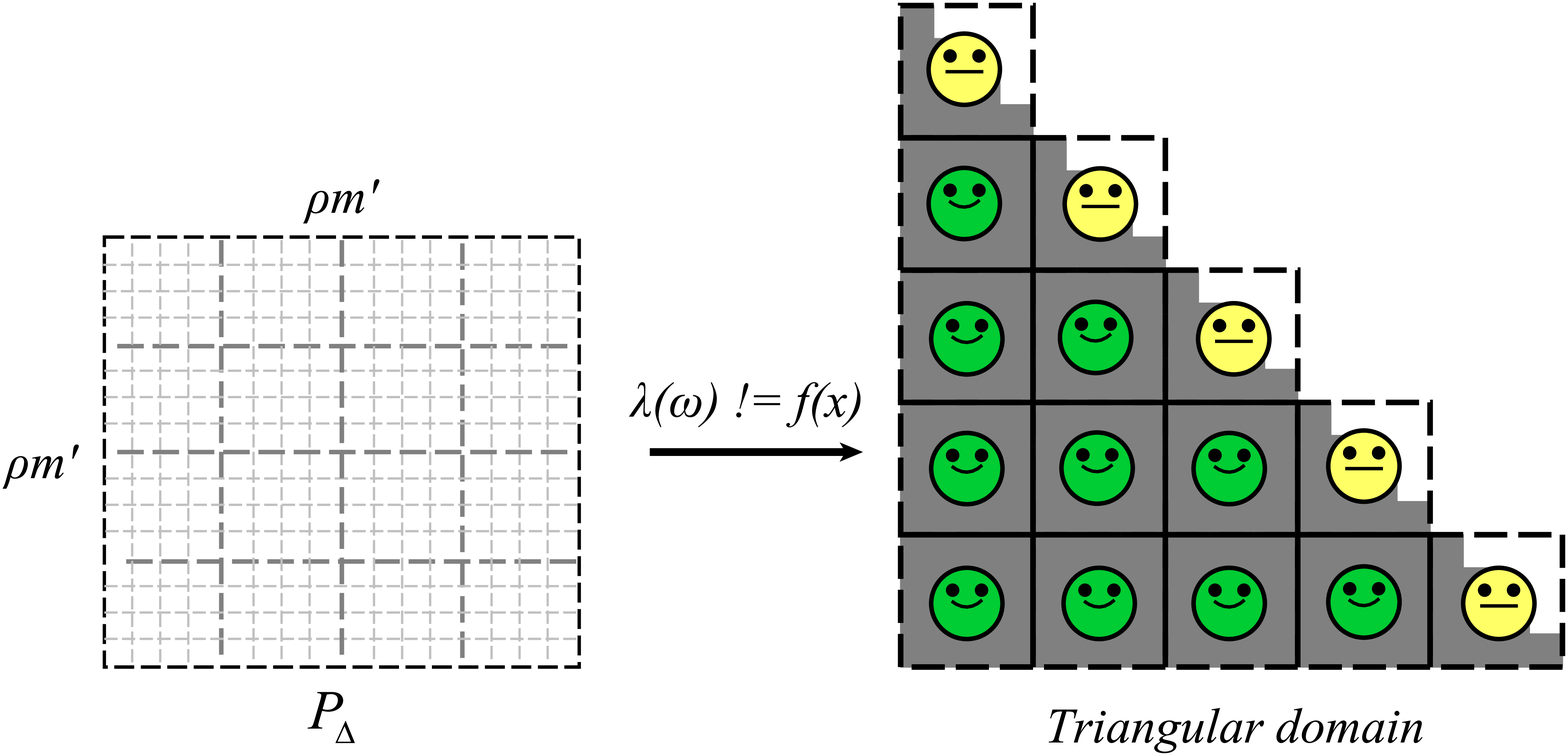}
\caption{parallel space $P_{\Delta}$ is sufficient to cover the problem domain.}
\label{fig_ltm_strategy}
\end{figure}

Spaces $P_{\Delta}$ and $D$ are both topologically equivalent, therefore the
map has to be at least a homeomorphism from block coordinates to
unique two-dimensional coordinates in the triangular data space.
\begin{theorem}
    There exists a non-linear homeomorphism $\lambda:\mathbb{Z}^1 \mapsto
    \mathbb{Z}^2$ that maps any block $B_{i,j} \in P_{\Delta}$ onto the TD
    class.
	%\begin{align}
    %\lambda(\omega) = (i,j) = \Big(\Big\lfloor\sqrt{\frac{1}{4} + 2\lambda} -
    %\frac{1}{2}\Big\rfloor, \lambda - i(i+1)/2\Big)
	%\label{eq_theorem}
	%\end{align}
\end{theorem}
\begin{proof}[Proof]
    Let $\omega$ be the linear index of a block, expressed as
	\begin{equation}
	\omega = \sum_{r=1}^{x}r
	\label{eq_packed_adyacency_sum_c}
	\end{equation}
    From expression (\ref{eq_matrix_indexing}), one can note that
    the index of the first data element in the row of the $\omega$-th linear block 
    corresponds to the sum in the range $[1,i]$, where $i$ is the row for
    the $\omega$-th block.  Similarly, the index of the first block from the next row
    is a sum in the range $[1,i+1]$. Therefore, for all $\omega$ values of the
    $i$-th row, their summation range is bounded as 
	\begin{align}
	\sum_{r=1}^{i}r & \le \lambda = \sum_{r=1}^{i+\epsilon}r < \sum_{r=1}^{i+1}r
    %\\\sum_{k=1}^{i}k & \le \sum_{k=1}^{x}k < \sum_{k=1}^{i+1}k
	\label{eq_packed_adyacency_sum}
	\end{align} 
    with $\epsilon < 1$. With this, we have that $x \in \mathbb{R}$ and most
    importantly that $i = \lfloor x\rfloor$. Since $\sum_{r=1}^{x}r = x(x+1)/2$, $x$
    is found by solving the second order equation $x^2 + x - 2\omega = 0$
    where the solution $x = \sqrt{1/4 + 2\omega} - 1/2$ allows the
    formulation of the homeomorphism
	\begin{equation}
        \lambda(\omega) = (i,j) = \Big(\Big\lfloor\sqrt{\frac{1}{4} + 2\omega} -
        \frac{1}{2}\Big\rfloor, \omega - i(i+1)/2\Big)
	    \label{eq_gmap}
    \end{equation}
    which is non-linear since $\exists \omega_1, \omega_2 \in P_{\Delta} :
    g(\omega_1 + \omega_2) \not= g(\omega_1) + g(\omega_2)$, \textit{e.g.,}
    $\omega_1 = 4, \omega_2 = 3 \implies g(7) = (3,1) \not= g(4) + g(3) =
    (4,1)$.
\end{proof}

If the diagonal is not needed, then $\lambda(\omega)$ becomes:
\begin{align}
\lambda(\omega) = (i,j) = \Big(\Big\lfloor\sqrt{\frac{1}{4} + 2\omega} +
\frac{1}{2}\Big\rfloor, \omega - i(i+1)/2\Big)
\label{eq_theorem_nodiag}
\end{align}

There are three important differences when comparing $\lambda(\omega)$ with the UTM
map \cite{AvrilGA12}: (1) $\lambda(\omega)$ maps in block-space and not in
thread-space as in UTM, allowing larger values of $n$, (2) thread organization
and locality is not compromised, making nearest-neighbors computations efficient
for shared memory and (3) $\lambda(\omega)$ uses fewer floating point operations than in
UTM since it uses a lower-triangular approach.

\subsection{Bounds on the improvement factor}
The reduction of unnecessary threads from $\mathcal{O}(n^2)$ to $\mathcal{O}(n)$
may suggest that the improvement could reach a factor of up to $2\times$. For
this to be possible, one would need to measure just the mapping stage, so that
necessary and unnecessary threads do similar amount of work, and assume that the mapping
function $\lambda(\omega)$ is as cheap as in the BB strategy. In the following
analysis we analyze the improvement factor considering a more realistic scenario
where $\lambda(\omega)$ has a higher cost than the BB map, due to the square root
computation involved.

The performance of $\lambda(\omega)$ strongly depends on the square root which in
theory costs $O(M(n))$ \cite{Ypma:1995:HDN:222504.222510} where $M(n)$ is the
cost of multiplying two numbers of $n$ digits.  Considering that real numbers
are represented by a finite number of digits (\textit{i.e.,} floating point
numbers with a maximum of $m$ digits), then all basic operations cost a fixed
amount of time, leading to a constant cost $M(m) = C_s \in \mathcal{O}(1)$. All
other computations are elemental arithmetic operations and can be taken as an
additional cost of $C_a \in \mathcal{O}(1)$. The total cost of $\lambda(\omega)$ is 
$\tau = C_s + C_a = \mathcal{O}(1)$ for each mapped thread.  On the other hand,
the BB strategy uses the identity map and checks for each thread if $B_j <= B_i$ 
in order to continue or be discarded, leading to a constant cost of $\beta \in
\mathcal{O}(1)$.
 It is indeed evident that $\beta$ is cheaper
than $\tau$, therefore $\tau = k\beta$ with a constant $k \ge 1$. The
improvement factor $I$ can expressed as 
\begin{equation}
I = \frac{\beta|P_{BB}|\rho^2}{\tau|P_{\Delta}|\rho^2} = \frac{2\beta
N_D^2}{\tau N_D^2 + \tau N_D} = \frac{2\beta\lceil n/\rho \rceil^2}{\tau(\lceil n/\rho
\rceil^2 + \lceil n/\rho \rceil)}
\label{eq_I}
\end{equation}
For large $n$ the result approaches to
\begin{equation}
    \lim_{n\to\infty} I = \frac{2\beta}{\tau}
\label{eq_I_bounded}
\end{equation}
Using the relation $\tau = k\beta$ in (\ref{eq_I_bounded}) we have $I \approx
2/k$ for large $n$ and since $k > 1$, the final range for $I$ becomes 
\begin{equation}
    0 < I < 2
\end{equation}
The parameter $k$ can be interpreted as the penalty factor of $\lambda(\omega)$,
where the lowest value is desired. In practice, a value $k \approx 1$ is too
optimistic. Given how actual GPU hardware works, one can expect that the cost of
the square root will dominate the $k$ parameter.

\section{Implementation}
\label{sec_implementation}
This section presents technical details on choosing a proper square root
implementation for $\lambda(\omega)$ as well as a general description of the 
related works chosen for performance comparison later on.  

\subsection{Choosing a proper square root} 
The performance of map $\lambda(\omega) = (i,j)$ depends, in great part, on how fast
the square root from eq.  (\ref{eq_gmap}) is computed.  Three versions of
$\lambda(\omega)$ have been implemented, each one using a different method for
computing the square root. Results from each test are computed as the
improvement factor with respect to the BB strategy.  The first implementation,
named $\lambda_X$, uses the default $sqrtf(x)$ function from CUDA C and it is the
simplest one.

The second implementation, named $\lambda_N$, computes the square root by
using three iterations of the Newton-Raphson method
\cite{Ypma:1995:HDN:222504.222510, Peelle:1974:TNS:585882.585889} which is 
available from the implementation of Carmack and Lomont.  This square root
implementation has proved to be effective for applications that allow small
errors.  The initial value used is the magic number ``0x5f3759df'' (this value
became known when 'Id Software' released Quake 3 source code back in the year
2005).  Adding a constant of $\epsilon = 10^{-4}$ to the
result of the square root can fix approximation errors in the range $N \in [0,
30720]$.

The third implementation, named $\lambda_R$, uses the hardware implemented
reciprocal square root, $rsqrtf(x)$:
\begin{equation}
\sqrt{x} = \frac{x}{\sqrt{x}} = x\cdot rsqrtf(x) 
\end{equation}
In terms of simplicity, $\lambda_R$ is similar to $\lambda_X$, with the only
difference that it adds $\epsilon = 10^{-4}$ at the end to fix approximation
errors, just like in $g_2$.

For each implementation a performance improvement factor was obtained with
respect to BB strategy, by running a dummy kernel that computes the $i,j$
indices and writes the sum $i+j$ to a constant location in memory. It is
necessary to perform at least one memory access otherwise the compiler can
optimize the code removing part of the mapping cost.  Figure
\ref{fig_results_map} shows the improvement factor as $I = BB/\lambda(\omega)$ using
the three different implementations, running on three different Nvidia Kepler
GPUs; GTX 680, GTX 765M and Tesla K40. 
\begin{figure*}[ht!]
\centering
\begin{tabular}{cc}
\includegraphics[scale=0.7]{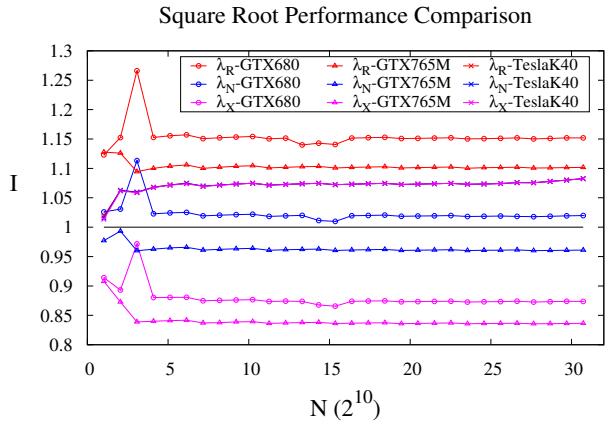}
\end{tabular}
\caption{Performance of the different square root strategies.}
\label{fig_results_map}
\end{figure*}

From the results, we observe that $\lambda_X$ is slower than BB when running
on the GTX 680 and GTX 765M, achieving $I_{680} \approx 0.87$ and $I_{765M}
\approx 0.83$, respectively.  For the same two GPUs, $\lambda_N$ achieves an
improvement of $I_{680} \approx 1.025$ and $I_{765M} \approx 0.96$ which is
practically the performance of BB. Lastly, $\lambda_R$ achieves
improvements of $I_{680} \approx 1.15$ and $I_{765M} \approx 1.1$. From these
results, we observe that using the inverse square root is the best option for
the GTX 680 and GTX 765M. For the case of the Tesla K40, we observe that all
three implementations achieve an improvement of $I_{K40} \approx 1.08$, allowing
to have the precision of $\lambda_X$ with the performance of $\lambda_R$.

\subsection{Implementing the other strategies}
The strategies from the relevant works, including the default one from CUDA,
were implemented as well; \textit{bounding-box (BB)}, \textit{rectangle-box
}(RB),\textit{recursive-partition} (REC) and \textit{upper-triangular-matrix}
(UTM) following the details provided by the authors \cite{Jung2008, AvrilGA12,
Ries:2009:TMI:1654059.1654069}.  To each implementation the following
restriction was added: the map cannot use any additional information that grows
as a function of $N$.  This means no auxiliary array such as lookup tables are
allowed, only small constants size data if needed.  The purpose of such
constraint is to guarantee that GPU memory is dedicated to the application problem.

For the \textit{bounding box} (BB) strategy, blocks above the diagonal are
discarded at runtime, without needing to compute a thread coordinate as it can
be done by checking if $B_i > B_j$ is true or not in the kernel. Threads
that got a true result from the conditional proceed to compute their global
coordinate and do the kernel work. The condition $i > j$ is still performed to
discard threads on blocks where $B_i = B_j$. It is important to note that this 
implementation of BB is faster than computing the thread coordinate first and
filtering afterwards.

The \textit{rectangular box} (RB) takes a sub-triangular portion
of the threads where $t_x > N/2$, rotates it \textit{counter-clock-wise} (CCW)
and places it above the diagonal to form a rectangular grid (see Figure
\ref{fig_other_methods}, left). 
\begin{figure}[ht!]
\centering
\includegraphics[scale=0.22]{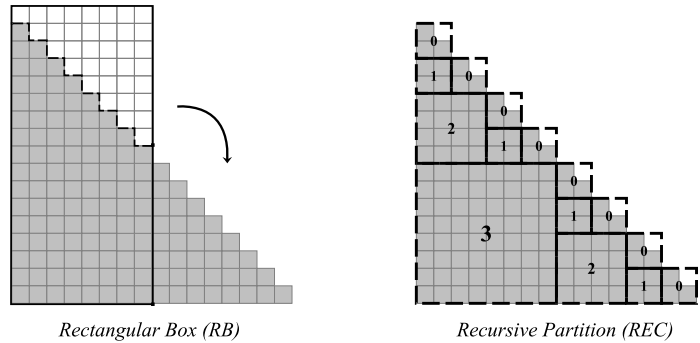}
    \caption{On the left, the \textit{rectangular-box} (RB). On the right, the
    \textit{recursive-partition} (REC).}
\label{fig_other_methods}
\end{figure}

The original work was actually a memory packing technique in data space, but the
principle can be applied to the parallel space as well.  For this, the lookup
texture is no longer used and instead the mapping coordinates are computed at
runtime.  All threads below the diagonal just need to map to $i = t_y-1$, while
$j$ remains the same, and threads in or above the diagonal map to $i = n - t_y -
1$,$j = n - i - 1$.  An important feature of the RB map is that the number of
unnecessary threads is asymptotically $\mathcal{O}(1)$.

The \textit{recursive partition} (REC) \cite{Ries:2009:TMI:1654059.1654069}
strategy was originally proposed for matrix inversion problem. In this strategy
the size of the problem is defined as $N = m2^k$ where $k$ and $m$ are positive
integers and $m$ is a multiple of the block-size $\rho$. The idea is to do a
binary \textit{bottom-up} recursion of $k$ levels (see Figure
\ref{fig_other_methods}, right), where the $i$-th level has half the number of
blocks of the $(i-1)$-th level, but with doubled linear size. This method
requires an additional pass for computing the blocks at the diagonal (level
$k=0$ is a special one).  More details of how the grid is built and how blocks
are distributed are well explained in \cite{Ries:2009:TMI:1654059.1654069}.  In
the original work, the mapping of blocks to their respective locations at each
level is achieved by using a lookup table stored in constant memory.  In this
case, the lookup table is discarded and instead the mapping is done at runtime. 

The \textit{upper-triangular mapping} (UTM) was proposed by Avril \textit{et.
al.} \cite{AvrilGA12} for performing efficient collision detection on the GPU.
Given a problem size $N$ and a thread index $k$, its unique pair $(a, b)$ is
computed as $a = \lfloor \frac{-(2n + 1) + \sqrt{4n^2 - 4n - 8k +1}}{-2}
\rfloor$ and $b = (a+1) + k - \frac{(a-1)(2n-a)}{2}$. The UTM strategy uses the
idea of mapping threads explicitly to the upper-triangular matrix, making it a
\textit{thread space} map.

\section{Performance Results}
\label{sec_experimental_results}
The experimental design consists of measuring the performance of $\lambda(\omega)$
and compare it against the \textit{bounding box} (BB), \textit{rectangular box}
(RB) \cite{Jung2008}, the \textit{recursive partition} (REC)
\cite{Ries:2009:TMI:1654059.1654069} and \textit{upper-triangular mapping} (UTM)
\cite{AvrilGA12}. Three tests are performed to each strategy; (1) the dummy
kernel, (2) EDM and (3) Collision detection. Test (1) just writes the $(i,j)$
coordinate into a fixed memory location. The purpose of the dummy kernel is to
measure just the cost of the strategy and not the cost of the application
problem.  Test (2) consists of computing the Euclidean distance matrix (EDM)
using four features, \textit{i.e.,} $(x,y,z,w)$ where all of the data is
obtained from global memory.  Test (3) consists of performing collision
detection of $N$ spheres with random radius inside a unit box. The goal of this
last test is to measure the performance of the kernels using a shared memory
approach.

The reason why  these tests were chosen is because they are simple enough to
study their performance from a GPU map perspective and use different memory
access paradigms such as global memory and shared memory.  Based on these
arguments, it is expected that the performance results obtained by these three
tests can give insights on what would be the behavior for more complex problems
that fall into one of the two memory access paradigms.  Furthermore, the tests
have been ran on three different GPUs in order to check if the results vary
under older or newer GPU architectures. The details on the maximum number of
\textit{simultaneous blocks} (\textit{sblocks}) for each GPU used are listed in
Table \ref{table_hardware}. 
\begin{table}[ht!]
\caption{Hardware used for experiments.}
\begin{center}
\begin{tabular}{|c|r|r|r|r|r|}
\hline
Device	&	Model            & Architecture  & Memory & Cores & \textit{sblocks} \\
\hline
$GPU_1$	&	Geforce GTX 765M & GK106         & 2GB    & 768   & 80 \\
$GPU_2$	&	Geforce GTX 680  & GK104         & 2GB    & 1536  & 128 \\
$GPU_3$ &	Tesla K40        & GK110         & 12GB   & 2880  & 240 \\
\hline
\end{tabular}
\end{center}
\label{table_hardware}
\end{table}

Performance results for the dummy kernel, 4D-EDM and collision detection in 1D/3D
are presented in Figure \ref{fig_all_results}. Graphic plots the
performance of all four strategies as different dashed line colors, while the
symbol type indicates which GPU was chosen for that result. The performance of
each mapping strategy is given in terms of its improvement factor $I$ with
respect to the $BB$ strategy (\textit{i.e.,} the black and solid horizontal line
fixed at $I=1$). Values that are located above the horizontal line represent
actual improvement, while curves that fall below the horizontal line represent a
slowdown with respect to the $BB$ strategy.
\begin{figure*}[ht!]
\centering
\begin{tabular}{cc}
\includegraphics[scale=0.46]{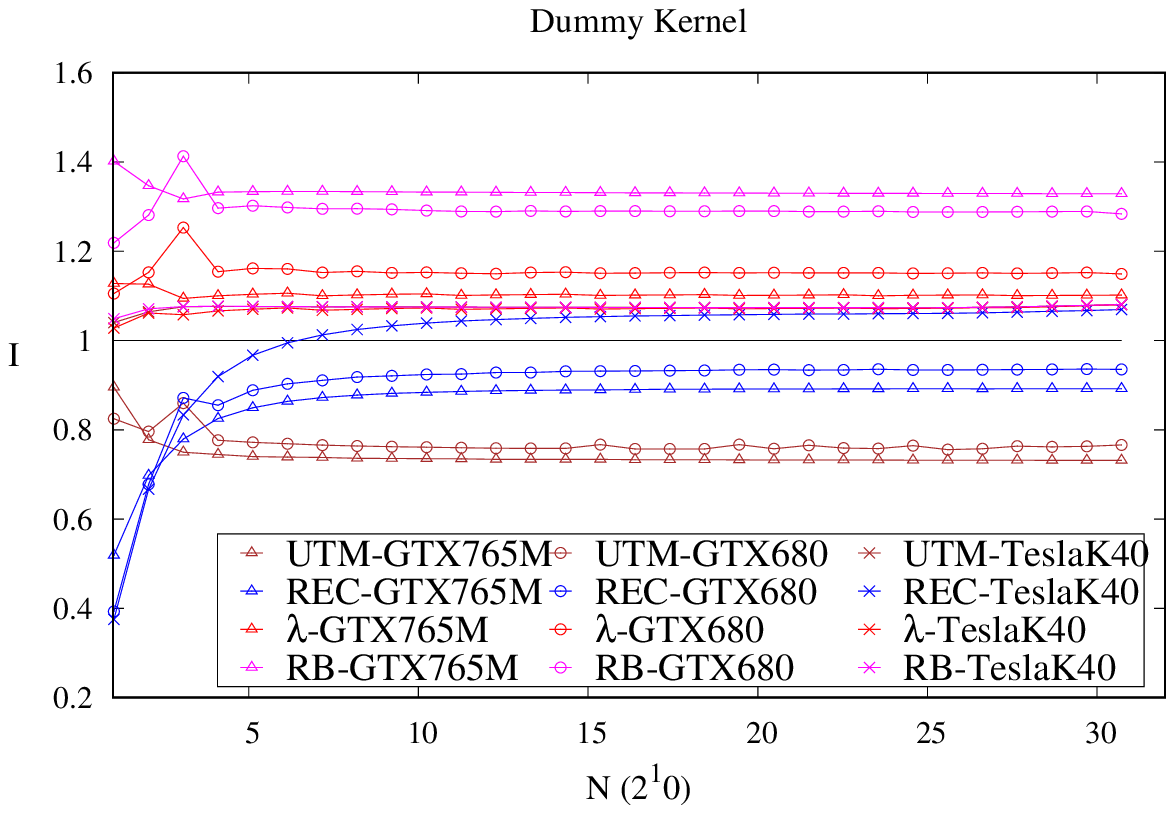} &
\includegraphics[scale=0.46]{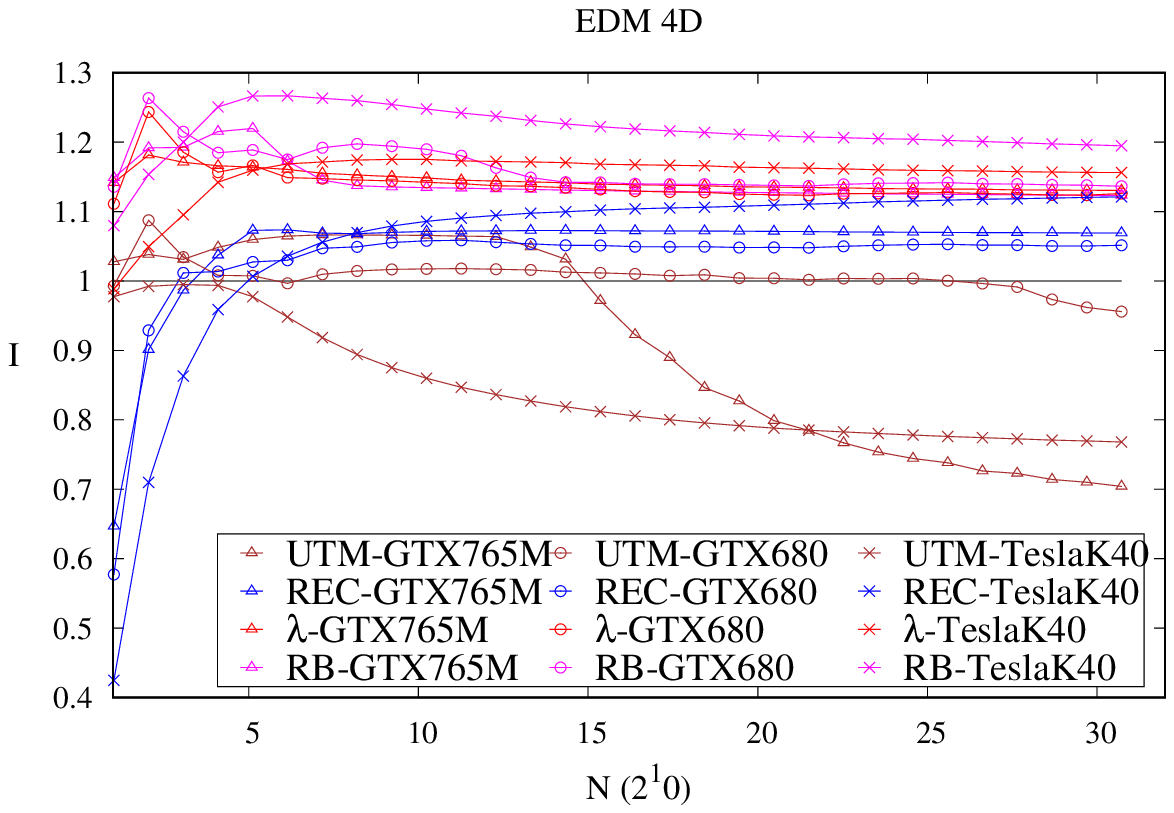}\\
\includegraphics[scale=0.46]{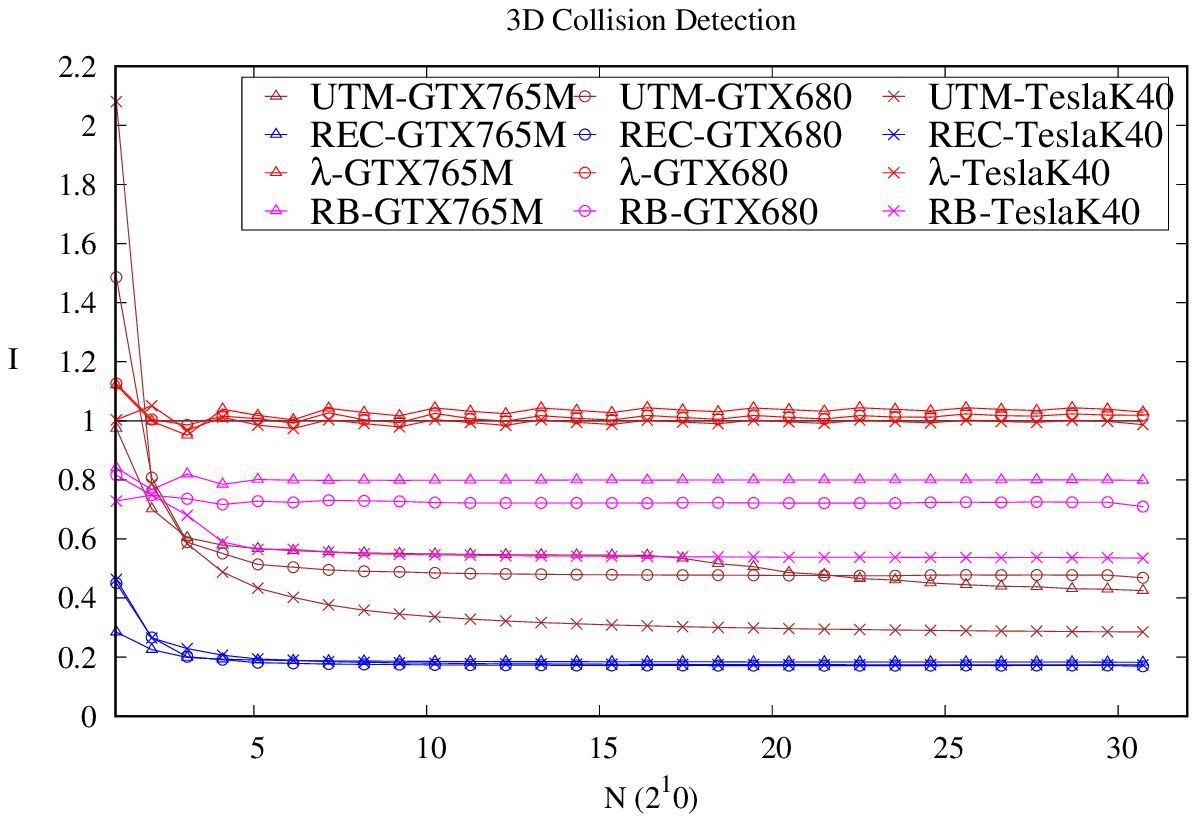} &
\includegraphics[scale=0.46]{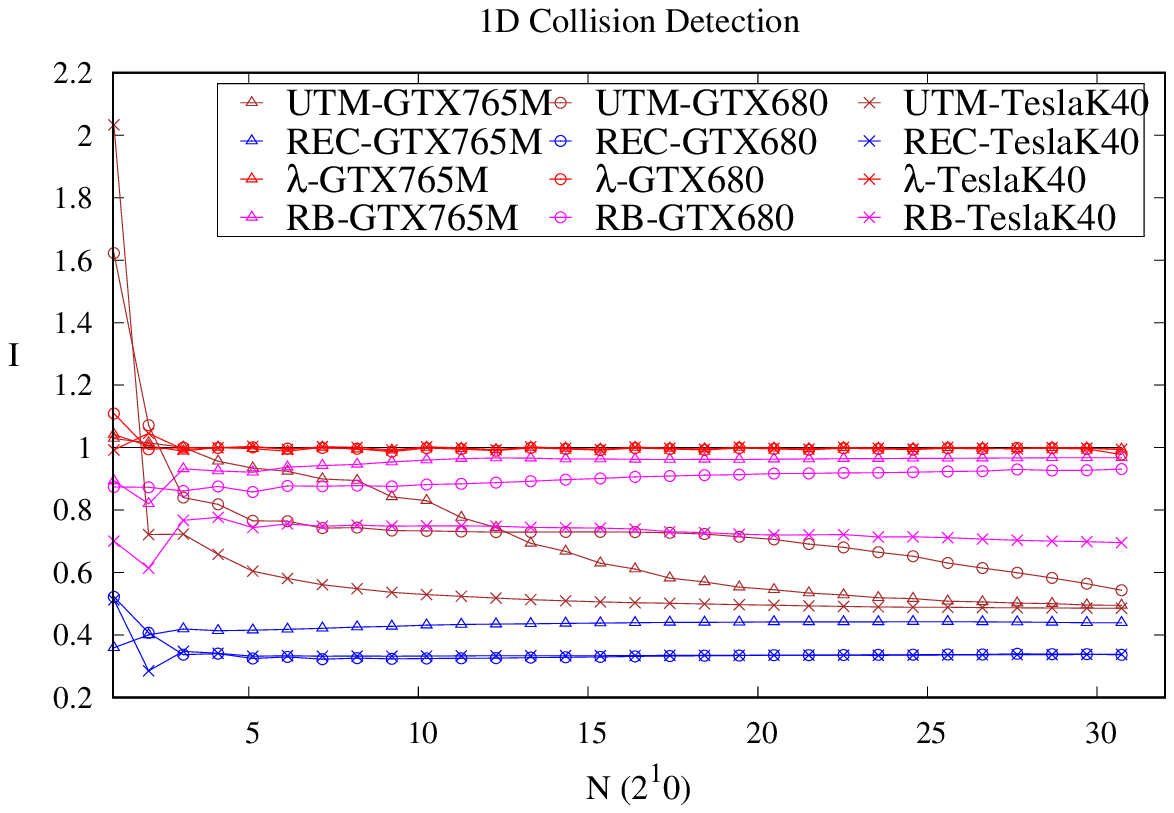}
\end{tabular}
\caption{Improvement factors for the dummy kernel, Euclidean distance matrix and
collision detection.}
\label{fig_all_results}
\end{figure*}
For the dummy kernel test the plots shows that the RB strategy is the fastest
one achieving up to 33\% of improvement with respect to BB when running on the
GTX 765M.  Map $\lambda(\omega)$ comes in the second place, achieving a stable
improvement of up to 18\% when running on the GTX 680.  The REC and UTM
strategies performed slower than BB for the whole range of $N$. We note that
this test running on the Tesla K40 does not show any clear performance
difference among the mapping strategies once $N > 15000$, as they all
converge to a 7\% of improvement with respect to BB.

%The EDM problem is solved by computing the Euclidean distance $d_{ij}$ for all
%pairs. The distance between a pair of elements $a_i$, $a_j$ is computed as:
%\begin{equation} f(i,j) = \sqrt{\sum_{k=1}^{d}({a_i^k} - {a_j^k})^2}
%\end{equation} Where $d$ is the number of features and superscript $k$
%specifies which feature of the element will be used.

For the EDM test, RB is again the fastest map achieving an improvement of up to
28\% with respect to BB when running on the Tesla K40 GPU.  In second place
comes $\lambda(\omega)$ with a stable improvement of 18\% and third the REC map with an
improvement that reaches up to 12\% for the largest $n$. The performance of the
UTM strategy was lower than BB and unstable for all GPUs. It is important to 
consider that UTM was designed to work in the range $n \in [0,3000]$, where
it actually does perform better than BB offering up to 4\% of improvement over
BB. For this test, performance differences among the strategies did manifest for
all GPUs, including the Tesla K40.

For the 3D collision detection test only $\lambda(\omega)$ manages to perform better than
BB, offering an improvement of up to 7\%. Indeed, the performance scenario
changes drastically in the presence of a different memory access pattern such as
shared memory; the RB strategy, which was the best in global memory, now
performs slower than BB.  The case is similar with the REC map which now
performs much slower than BB. It is important to mention that the UTM map is the only
strategy that cannot use a 2D shared memory pattern because the mapping works in
\textit{linear thread space}. At low $n$, UTM achieves a 100\% of improvement
because of the different memory approach used, thus it cannot be taken as a
practical improvement.  Furthermore, as $n$ grows, its performance is over passed by
the rest of the strategies that use shared memory.  For the case of 1D collision
detection, results are not so beneficial for the mapping strategies and in the
case of $\lambda(\omega)$ it is in the limit of being an improvement. The 1D results
show that in low dimensions the benefits of a mapping strategy can be not as
good as in higher ones.

\section{Extension to 3D Tetrahedrons}
\label{sec_extension3d}
In this section, the potential benefit of $\lambda(\omega)$ is considered for 3D
cases as an extension of the 2D approach. The three-dimensional analog of the 2D
triangle corresponds to the 3D discrete tetrahedron, which may be defined by $n$
triangular structures stacked and aligned at their right angle, where the $r$-th
trianglular layer contains $T_r^{2D} = r(r+1)/2$ elements. The total number of
elements for the full structure can be expressed in terms of the $n$ layers:
\begin{equation}
    T_n = \sum_{r=1}^{n} T_r^{2D} = \sum_{r=1}^{n}{\frac{r(r+1)}{2}}
\end{equation}
The sequence corresponds to the tetrahedral numbers, which can be
defined by 
\begin{equation}
    T_n = {{n+2}\choose{3}} = \sum_{r=1}^{n} T_r^{2D} = \frac{n(n+1)(n+2)}{6} \end{equation}

Similar to the 2D case, a canonical map $f(x) = x$ with a box-shaped parallel
space leads to an innefficiency as the number of unnecessary threads is in the
order of $O(n^3)$ as illustrated by Figure
\ref{fig_bounding-box-tetrahedron}.
\begin{figure}[ht!]
    \includegraphics[scale=0.15]{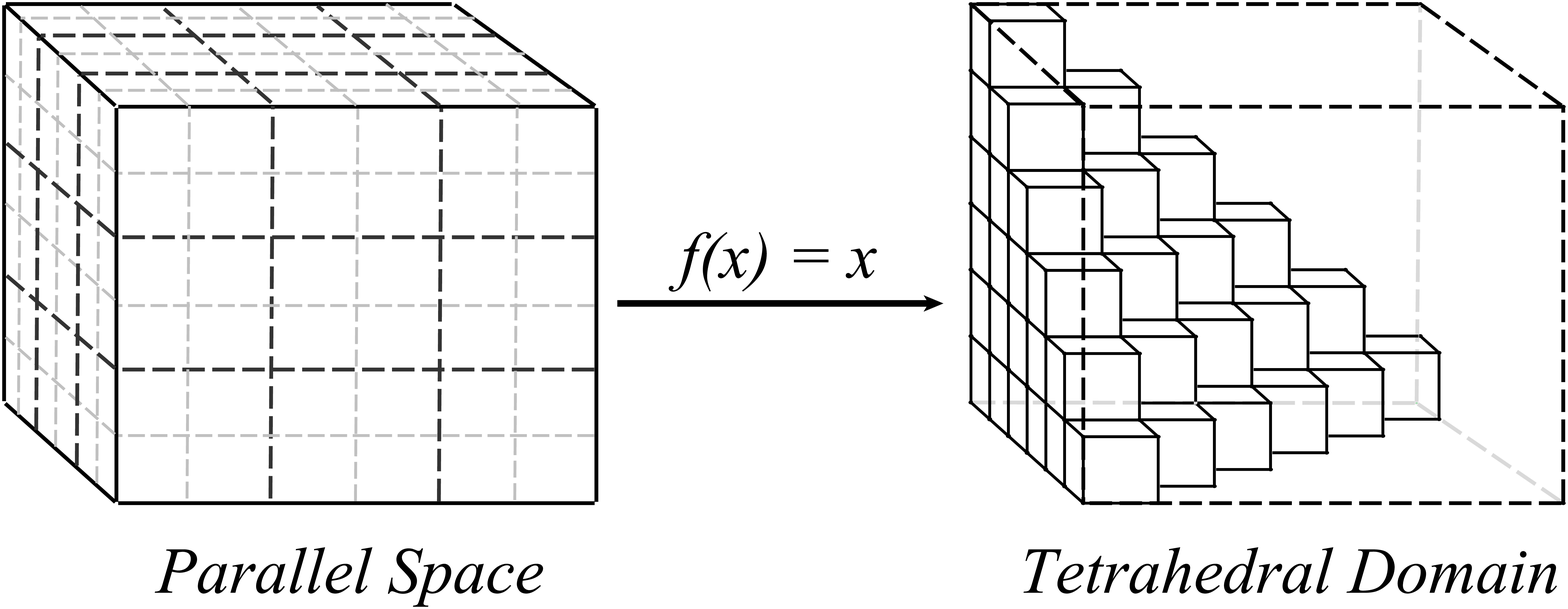}
    \centering
    \caption{A bounding-box approach produces $O(n^3)$ unnecessary
    threads.}
    \label{fig_bounding-box-tetrahedron}
\end{figure}

A more efficient approach can be formulated by considering how block indices can
map onto the tetrahedron.  More precisely, it is possible to redefine $\lambda$ as 
a map $\lambda(\omega):\mathbb{N} \rightarrow \mathbb{N}^3$ that works on the 
tetrahedral structure without loss of parallelism (see Figure \ref{fig_tetrahedron-map}). 
\begin{figure}[ht!]
    \includegraphics[scale=0.18]{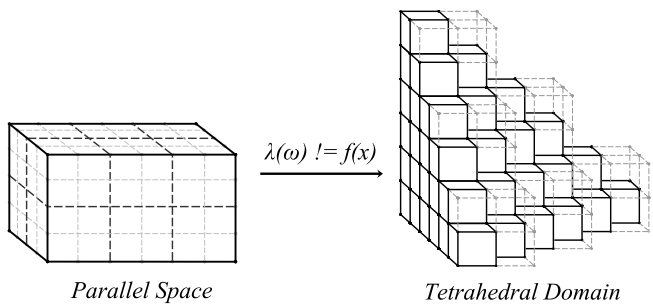}
    \centering
    \caption{Blocks of threads can use an extended version of $\lambda(\omega)$ to map from parallel space
    (left) onto the tetrahedral structure (right).}
    \label{fig_tetrahedron-map}
\end{figure}

The approach takes advantage of the fact that when using a linear enumeration of
blocks on the tetrahedron, the $\omega$ index of the first element of a 2D
triangular layer corresponds to a tetrahedral number $T_x$. Similar to the
previous two-dimensional row analysis, now that data elements that reside in the same
layer obey the following property
\begin{equation}
    \sum_{r=1}^{k} {r(r+1)/2} < \omega = \sum_{r=1}^x {r(r+1)/2} <
    \sum_{r=1}^{k+1} {r(r+1)/2}.
\end{equation}
When expressing $\omega$ as the tetrahedral number
\begin{equation}
    \omega = \sum_{r=1}^x {r(r+1)/2} = \frac{x(x+1)(x+2)}{6}.
\end{equation}
the $k$ component of the $(i,j,k)$ coordinate of block $\omega$ can be obtained
by solving the third order equation
\begin{equation}
    x^3 + 3x^2 + 2x -6\omega = 0
\end{equation}
and extracting the integer part of the root
\small
\begin{equation}
    x   = \frac{\sqrt[3]{\sqrt{729\omega^2 - 3} + 27\omega}}{3^{2/3}} +
    \frac{1}{\sqrt[3]{3} \sqrt[3]{ \sqrt{729\omega^2 - 3} + 27\omega}} - 1
\end{equation}
\normalsize
Once the value $k = \lfloor x \rfloor$ is computed, the 
$\omega_{2D}$ linear coordinate 
\begin{equation}
    \omega_{2D} = \omega - T_k 
\end{equation}
can be obtained as well, where $T_k = k(k+1)(k+2)/6$ is the tetrahedral number for the recently computed
$k$ value. With $\omega_{2D}$ computed, the $i$ and $j$ values of the block can
be computed using the two-dimensional version of $\lambda(\omega)$. Combining all three
sub-results, the tetrahedral map $\lambda(\omega)$ becomes
\begin{equation}
    \lambda(\omega) \mapsto (i,j,k) = \Big(\omega_{2D} - T_y^{2D}, \Big\lfloor\sqrt{\frac{1}{4} +
    2\omega_{2D}} - \frac{1}{2}\Big\rfloor, \lfloor v \rfloor\Big)
\end{equation}.

The blocks in parallel space can be organized on a cubic grid of side $\lceil
\sqrt[3]{T_n} \rceil$ in order to balance the number of elements on each
dimension, producing $n^2 \rho^3 \in o(n)$ unnecessary threads.
The potential improvement factor of the block-space map with respect to the
bounding box is 
\begin{equation}
I = \frac{\alpha n^3/\rho^3}{\gamma T_n/\rho^3} = \frac{6\alpha n^3}{\gamma (n^3 + 3n^2 + 2n)}
\end{equation}
where $\alpha$ is the cost of computing the block coordinate using the box
approach, while $\gamma$ is the cost of mapping blocks onto the tetrahedral map. 
In the infinite limit of $n$, the potential improvement becomes
\begin{equation}
    I_{n\to\infty} \sim \frac{6\alpha}{\gamma}
\end{equation}
and tells that in theory the tetrahedral map could be close to $6\times$ more
efficient for large $n$. However, such improvement can only be possible if $\gamma \sim
\alpha$, \textit{i.e.,} if the square and cubic roots can be computed fast
enough to be comparable to the cost of the bounding-box approach, which is very
unlikely to happen in practice. Nevertheless, there still can be valuable
improvement as long as $\gamma < 6\alpha$.

\section{Conclusions}
\label{sec_conclusions}
%The proposed GPU map, $\lambda(\omega)$, reduces the number of unnecessary threads to
%a marginal number of $o(N^2)$ and works in \textit{block space}, not affecting
%thread organization. Due to the block-space approach, the range of problem sizes 
%reaches is $N \in [0,30720]$, which is ten times the range achieved in a
%thread-space approach \cite{AvrilGA12}. It is important to mention that the
%problem size must be large enough to generate more blocks than the number of
%blocks the GPU can handle in parallel. Assuming a warp of 32 threads, a value of
%$N \ge 800$ already generates more than double the number of blocks in parallel
%for any of the GPUs we used in this work. For smaller values of $N$, the BB
%strategy will still remain useful.

The mapping technique studied in this work may be a useful optimization for
GPU solutions of parallelizable problems that have a triangular or tetrahedral domain.  The map
proposed in this work, namely $\lambda(\omega)$, proved to be the fastest strategy
(with $7\%$ of improvement over the \textit{bounding-box}) when dealing with
shared memory scenarios that require locality, which is a computational pattern
often found in scientific and engineering problems where information of nearest neighbors 
is accessed. Since the map is performed in block-space, thread
locality is not compromised allowing to integrate the optimization with other
optimization techniques such as block-level data re-organization or coalesced
data sharing among threads. The implementation of $\lambda(\omega)$ is short in
code and totally detached from the problem, making it easy to adopt it as a
self-contained function with no side effects.

The performance of the map varies depending on which GPU is used and on the way
the map computations are performed. It is of high interest to study even more
optimized square root routines as they have a great impact
on the performance of $\lambda(\omega)$.

Extensions to the 3D tetrahedron are worth inspecting in more detail as they
have the potential of being up to $6\times$ more efficient regarding parallel
space. This improvement can traduce to a performance increase only if the cost
of cubic and square root computations is low enough to the same order of the
bounding box approach. Although tetrahedral domain problems are less frequently
found, they are still important in science when solving triplet-interaction
n-body problems. 
An interesting approach for the implementation of the 3D
map would be to reconsider relaxing the condition of allowing extra data and 
introduce a type of \textit{succint lookup table} of $o(T_n)$ combined with 
coordinate computations in order to balance and overlap the use of numerical and
memory operations.

\section*{Acknowledgment} This work was supported by the FONDECYT project
N$^o$ 3160182 and the Nvidia GPU Research Center from University of Chile. 
\bibliographystyle{elsarticle-num}
\bibliography{linearmap}

\begin{thebibliography}{10}
\expandafter\ifx\csname url\endcsname\relax
  \def\url#1{\texttt{#1}}\fi
\expandafter\ifx\csname urlprefix\endcsname\relax\def\urlprefix{URL }\fi
\expandafter\ifx\csname href\endcsname\relax
  \def\href#1#2{#2} \def\path#1{#1}\fi

\bibitem{4490127}
J.~Owens, M.~Houston, D.~Luebke, S.~Green, J.~Stone, J.~Phillips, Gpu
  computing, Proceedings of the IEEE 96~(5) (2008) 879--899.
\newblock \href {http://dx.doi.org/10.1109/JPROC.2008.917757}
  {\path{doi:10.1109/JPROC.2008.917757}}.

\bibitem{Nickolls:2010:GCE:1803935.1804055}
J.~Nickolls, W.~J. Dally, The gpu computing era, IEEE Micro 30~(2) (2010)
  56--69.

\bibitem{navhitmat2014}
C.~A. Navarro, N.~Hitschfeld-Kahler, L.~Mateu, A survey on parallel computing
  and its applications in data-parallel problems using {GPU} architectures,
  Commun. Comput. Phys. 15 (2014) 285--329.

\bibitem{nvidia_cuda_guide}
Nvidia-Corporation, Nvidia CUDA C Programming Guide (2016).

\bibitem{opencl08}
{Khronos OpenCL Working Group}, The OpenCL Specification, version 1.0.29 (8
  December 2008).

\bibitem{5695222}
D.~Man, K.~Uda, H.~Ueyama, Y.~Ito, K.~Nakano, Implementations of parallel
  computation of euclidean distance map in multicore processors and gpus, in:
  Networking and Computing (ICNC), 2010 First International Conference on,
  2010, pp. 120--127.
\newblock \href {http://dx.doi.org/10.1109/IC-NC.2010.55}
  {\path{doi:10.1109/IC-NC.2010.55}}.

\bibitem{Li:2010:CME:1955604.1956601}
Q.~Li, V.~Kecman, R.~Salman, A chunking method for euclidean distance matrix
  calculation on large dataset using multi-gpu, in: Proceedings of the 2010
  Ninth International Conference on Machine Learning and Applications, ICMLA
  '10, IEEE Computer Society, Washington, DC, USA, 2010, pp. 208--213.

\bibitem{Man:2011:GIC:2117688.2118809}
D.~Man, K.~Uda, Y.~Ito, K.~Nakano, A gpu implementation of computing euclidean
  distance map with efficient memory access, in: Proceedings of the 2011 Second
  International Conference on Networking and Computing, ICNC '11, IEEE Computer
  Society, Washington, DC, USA, 2011, pp. 68--76.

\bibitem{AvrilGA12}
Q.~Avril, V.~Gouranton, B.~Arnaldi, Fast collision culling in large-scale
  environments using gpu mapping function, in: EGPGV, 2012, pp. 71--80.

\bibitem{kepner2011graph}
J.~Kepner, J.~Gilbert,
  \href{http://books.google.com.hk/books?id=BnezR\_6PnxMC}{Graph Algorithms in
  the Language of Linear Algebra}, Software, Environments, Tools, Society for
  Industrial and Applied Mathematics, 2011.
\newline\urlprefix\url{http://books.google.com.hk/books?id=BnezR\_6PnxMC}

\bibitem{ConwaysLife}
M.~Gardner, {The fantastic combinations of John Conway's new solitaire game
  ``life''}, Scientific American 223 (1970) 120--123.

\bibitem{Ries:2009:TMI:1654059.1654069}
F.~Ries, T.~De~Marco, M.~Zivieri, R.~Guerrieri, Triangular matrix inversion on
  graphics processing unit, in: Proceedings of the Conference on High
  Performance Computing Networking, Storage and Analysis, SC '09, ACM, New
  York, NY, USA, 2009, pp. 9:1--9:10.

\bibitem{springerlink_gustavson}
F.~Gustavson, New generalized data structures for matrices lead to a variety of
  high performance algorithms, in: R.~Wyrzykowski, J.~Dongarra, M.~Paprzycki,
  J.~Wasniewski (Eds.), Parallel Processing and Applied Mathematics, Vol. 2328
  of Lecture Notes in Computer Science, Springer Berlin / Heidelberg, 2006, pp.
  418--436.

\bibitem{DBLP:journals/corr/abs-1108-5815}
R.~Yokota, L.~A. Barba, Fast n-body simulations on {{GPU}s}, CoRR
  abs/1108.5815.

\bibitem{Bedorf:2012:SOG:2133856.2134140}
J.~B{\'e}dorf, E.~Gaburov, S.~Portegies~Zwart,
  \href{http://dx.doi.org/10.1016/j.jcp.2011.12.024}{A sparse octree
  gravitational n-body code that runs entirely on the {GPU} processor}, J.
  Comput. Phys. 231~(7) (2012) 2825--2839.
\newblock \href {http://dx.doi.org/10.1016/j.jcp.2011.12.024}
  {\path{doi:10.1016/j.jcp.2011.12.024}}.
\newline\urlprefix\url{http://dx.doi.org/10.1016/j.jcp.2011.12.024}

\bibitem{Ivanov:2007:NPT:1231091.1231100}
L.~Ivanov, \href{http://dl.acm.org/citation.cfm?id=1231091.1231100}{The n-body
  problem throughout the computer science curriculum}, J. Comput. Sci. Coll.
  22~(6) (2007) 43--52.
\newline\urlprefix\url{http://dl.acm.org/citation.cfm?id=1231091.1231100}

\bibitem{Ying:2011:GDD:2065356.2065583}
Z.~Ying, X.~Lin, S.~C.-W. See, M.~Li, Gpu-accelerated dna distance matrix
  computation, in: Proceedings of the 2011 Sixth Annual ChinaGrid Conference,
  CHINAGRID '11, IEEE Computer Society, Washington, DC, USA, 2011, pp. 42--47.

\bibitem{Jung2008}
J.~H. Jung, D.~P. O’Leary, Exploiting structure of symmetric or triangular
  matrices on a gpu, Tech. rep., University of Maryland (2008).

\bibitem{Peelle:1974:TNS:585882.585889}
H.~A. Peelle, To teach {N}ewton's square root algorithm, SIGAPL APL Quote Quad
  5~(4) (1974) 48--50.

\bibitem{Ypma:1995:HDN:222504.222510}
T.~J. Ypma, Historical development of the {N}ewton-{R}aphson method, SIAM Rev.
  37~(4) (1995) 531--551.

\bibitem{DBLP:conf/hpcc/NavarroH14}
C.~A. Navarro, N.~Hitschfeld,
  \href{http://dx.doi.org/10.1109/HPCC.2014.64}{{GPU} maps for the space of
  computation in triangular domain problems}, in: 2014 {IEEE} International
  Conference on High Performance Computing and Communications, 6th {IEEE}
  International Symposium on Cyberspace Safety and Security, 11th {IEEE}
  International Conference on Embedded Software and Systems, {HPCC/CSS/ICESS}
  2014, Paris, France, August 20-22, 2014, 2014, pp. 375--382.
\newblock \href {http://dx.doi.org/10.1109/HPCC.2014.64}
  {\path{doi:10.1109/HPCC.2014.64}}.
\newline\urlprefix\url{http://dx.doi.org/10.1109/HPCC.2014.64}

\bibitem{NylandHarris2007}
L.~Nyland, M.~Harris, J.~Prins, {Fast N-Body Simulation with CUDA}, in:
  H.~Nguyen (Ed.), GPU Gems 3, Addison Wesley Professional, 2007, Ch.~31, pp.
  677--795.

\end{thebibliography}
\end{document}